\RequirePackage[l2tabu, orthodox]{nag}

\documentclass[english,utf8,biblatex]{lni}
\bibliography{lni-paper-example-de}
\pdfoutput=1 

\usepackage{booktabs}

\usepackage[math]{blindtext}
\usepackage{mwe}

\iffalse
\AtEveryBibitem{%
  \ifentrytype{article}{%
  }{%
    \clearfield{doi}%
    \clearfield{issn}%
    \clearfield{url}%
    \clearfield{urldate}%
  }%
  \ifentrytype{inproceedings}{%
  }{%
    \clearfield{doi}%
    \clearfield{issn}%
    \clearfield{url}%
    \clearfield{urldate}%
  }%
}
\fi

\usepackage{amsmath}
\usepackage{amsfonts}
\usepackage{amsthm}

\usepackage{booktabs}

\usepackage{needspace}
\usepackage{paralist}
\usepackage{scrextend}          
\usepackage{graphicx}
\usepackage{subcaption}
\usepackage{paralist}

\usepackage[textsize=footnotesize,backgroundcolor=green!10]{todonotes}
\usepackage[final,notref,notcite]{showkeys}

\newcommand{\bigO}{\ensuremath{\operatorname{O}}}

\newcommand{\tail}[1]{\ensuremath{\text{tail}(#1)}}
\newcommand{\head}[1]{\ensuremath{\text{head}(#1)}}
\newcommand{\arrival}[1]{\ensuremath{\text{arr}(#1)}}
\newcommand{\arrivalingraph}[2]{\ensuremath{\text{arr}_{#2}(#1)}}

\newcommand{\gamenamelong}{\probname{Robust Connection Game}}
\newcommand{\gamename}{\probname{RCG}}

\newcommand{\st}{s.\,t.}
\newcommand{\wlogg}{w.\,l.\,o.\,g.}
\newcommand{\Wlogg}{W.\,l.\,o.\,g.}

\theoremstyle{plain}
\newtheorem{theorem}{Theorem}[section]
\newtheorem{lemma}[theorem]{Lemma}

\newtheorem{observation}[theorem]{Observation}

\theoremstyle{definition}

\newcommand{\probname}[1]{\textsc{#1}}

\newcommand{\prob}[6]{%
  \needspace{3\baselineskip}
  \begin{quote}
    \begin{labeling}{#6}%
      \setlength\itemsep{-.4ex}
    \item[#1]
    \item[\emph{#2}]#3
    \item[\emph{#4}]#5
    \end{labeling}%
  \end{quote}%
}

\usepackage{complexity}

\usepackage{todonotes}

\usepackage{etoolbox}
\newcommand{\appsymb}{$\bigstar$}
\newcommand{\appref}[1]{\hyperref[proof:#1]{\appsymb}}

\newif\ifarxiv
\arxivtrue

\newcommand{\appendixsection}[1]{%
	\ifarxiv{}\else{}
	\gappto{\appendixProofText}{\section{Additional Material for Section~\ref{#1}}\label{app:#1}}
	\fi{}
}

\newcommand{\toappendix}[1]{%
	\ifarxiv{}#1\else{}
	\gappto{\appendixProofText}
	{{
			#1
	}}
	\fi{}
}

\newcommand{\appendixproof}[2]{%
	\ifarxiv{}\begin{proof}#2\end{proof}\else{}\gappto{\appendixProofText}
	{
		\subsection{Proof of \cref{#1}}\label{proof:#1}
		#2
	}
	\fi{}
}

\begin{document}
\title{Network Navigation with Online Delays is \PSPACE-complete}
\author[Depian, Kern, R\"oder, Terziadis and Wallinger]
{Thomas Depian\footnote{[e11807882$\vert$e11904675$\vert$sebastian.roeder]@student.tuwien.ac.at, TU Wien, Austria} \and Christoph Kern$^1$ \and Sebastian Röder$^1$ \and Soeren Terziadis\footnote{[sterziadis$\vert$mwallinger]@ac.tuwien.ac.at, Algorithms \& Complexity Group, TU Wien, Austria} \and Markus Wallinger$^2$}
\startpage{1} 
\editor{Gesellschaft für Informatik}    
\booktitle{SKILL 2023} 
\yearofpublication{2023}
\maketitle

\begin{abstract}
  \looseness=-1
  In public transport networks disruptions may occur and lead to travel delays.
  It is thus interesting to determine whether a traveler can be resilient to delays that occur unexpectedly, ensuring that they can reach their destination in time regardless.
  We model this as a game between the traveler and a delay-introducing adversary.
  We study the computational complexity of the problem of deciding whether the traveler has a winning strategy in this game.
  Our main result is that this problem is \PSPACE-complete.


\end{abstract}

\begin{keywords}
  temporal paths \and network navigation \and robust connections
\end{keywords}

\section{Introduction}
According to Destatis, the total distance traveled by individuals in Germany using public transport in 2022 amounted to 99 billion kilometers~\cite{statistisches_genesis_2022}.
Finding the best public transport route between a starting point and a destination is a well-researched topic and there are well-known algorithms and datastructures for computing such routes~\cite{BastDGMPSWW16}.
Once the route has been determined, however, the traveler may encounter additional challenges.
According to rail company Deutsche Bahn, more than 13\,\% of the stops of their long-distance trains were not within 15 minutes of the schedule in April 2023~\cite{DeutscheBahn_2023}.
Travelers thus regularly need to board delayed trains, perhaps causing them to miss connecting trains later on.
Therefore, it is 
an interesting problem to find out whether a traveler can be resilient to such delays.
The problem formulation we are interested in here models the following question: Starting from point~$s$, is it possible to reach point~$z$ in time $t$ 
where a delay of a connecting train may occur unexpectedly at any changeover?

The above question can be modeled as a game between a traveler and a public transport company.
In each round of the game, the traveler arrives at some station in the public transport network.
Then the public transport company announces delays of the connections.
The traveler then decides on which connection to take next and so on, until either the traveler reaches the destination $z$ or the time is up. 
To model realistic scenarios we impose a delay-budget constraint on the public transport company, that is, the announced delays may sum up to at most some fixed budget.  
Whether the traveler is resilient to delays is then equivalent to whether they 
have a winning strategy, that is, whether it is possible to reach the destination in time regardless of which delays are being announced in each round.
We call the resulting decision problem \gamenamelong.

For very important appointments it may be useful to check beforehand whether they can be reached even with potential delays.
That is, we want to decide \gamenamelong\ computationally.
Thus it is interesting to know its computational complexity.
That is, how complex is it to decide, given the schedule of the network, the start and destination as well as the arrival time, whether the traveler has a winning strategy?

We show that \gamenamelong\ can be solved with space bounded polynomially in the input length, that is, \gamenamelong\ is contained in \PSPACE\ (see \cref{sec:PSPACE-containment}).
This in particular implies that it can be solved in time exponential in the input length.
However, we also prove that \gamenamelong\ is \PSPACE-hard, that is, every problem in \PSPACE\ can be reduced to \gamenamelong\ in polynomial time (see \cref{sec:PSPACE-hardness}).
This makes it unlikely that the problem can be solved efficiently in general.
While a negative result, our reduction highlights several features of the networks that we exploit in showing hardness, see the conclusion in \cref{sec:conclusion}.
It may be worthwhile to check to which extent these features occur in real-world networks and, if not, whether their absence can be exploited to obtain efficient algorithms.
\ifarxiv{}\fi

\paragraph{Related work.}
We model \gamenamelong\ on so-called temporal graphs, that is, graphs in which edges are equipped with time information such as their starting time and traversal time.
Routing on temporal graphs was to our knowledge first explored by \citet{Berman96} and has in recent years gained considerable attention.
Robustness of temporal connectivity was herein mostly studied with respect to deletion of a bounded number of time arcs 
or vertices
, see, \eg, the overview by~\citet{FuchsleMNR22STACS}.
In terms of delays in our context, we are aware of two works:

First, \citet{FuchsleMNR22STACS} study the problem of finding one route that is robust to a bounded number of unit delays, regardless of when they occur.
That is, checking whether there \emph{exists} a route that is feasible \emph{for all} delays.
In this model, the authors assume that the traveler never reconsiders the rest of the route.
The traveler thus potentially foregoes better connections that open up after they experienced some delays already.
The authors then study the complexity of finding such routes and whether efficient algorithms can be obtained if the number of delays or the network topology is restricted.

In the second work, \citet{FuchsleMNR22SAND} study finding \emph{for all} possible delays whether there \emph{exists} a delay-tolerant route.
In particular, they study the case where the delays may occur unexpectedly and model the resulting problem as a game similar to what we do here.
However, there is a crucial difference: In \citeauthor{FuchsleMNR22SAND}'s model a delay of a connection may be announced only \emph{during} the time in which the traveler takes the connection.  
In contrast, in our model the delays are announced \emph{before} the traveler decides on the next connection to take.
We would argue that both models are relevant and thus we close a gap in the literature.
It also turns out that this seemingly small difference has an immense effect on the computational complexity: In \citeauthor{FuchsleMNR22SAND}'s game model, deciding whether there is a winning strategy is polynomial-time solvable and only becomes \PSPACE-hard if we require the route taken by the traveler to be a path (that is, no vertex is traversed twice).
In contrast, in our model the problem is \PSPACE-hard without additional requirements on the route.



\section{Preliminaries}
\looseness=-1
In this paper, we work with temporal graphs, which build on top of static graphs.
A \emph{static graph} is a (directed) graph $G_s = (V_s, A_s)$ consisting of a set of \emph{vertices} $V_s$ and \emph{arcs} $A_s \subseteq V_s \times V_s$. 
We denote an \emph{arc} $a$ as a tuple $a = (u, v)$, and call $\tail{a} = u$ the \emph{tail} vertex, and $\head{a} = v$ the \emph{head} vertex of $a$.
The graph $G_s$ contains \emph{multi-arcs}, if there are two distinct arcs $a = (u, v), a' = (u, v) \in A_s$ that have the same vertices $u$ and $v$ as their head and tail.
We allow multi-arcs but we prohibit self-loops, \ie, arcs of the form $(u, u)$ for $u \in G_s$.
A \emph{walk} $W = (v_1, \dots, v_k)$ in $G_s$ is a sequence of $k$, not necessarily pairwise distinct, vertices such that we have $(v_i, v_{i + 1}) \in A_s$ for $1 \leq i < k$.
A walk is a \emph{path} if the vertices are pairwise distinct.

\paragraph{Temporal Graphs.}
A \emph{temporal graph} $G = (V, E)$ is a static directed graph where we replace arcs with temporal arcs. 
We denote a \emph{temporal arc} $e$ as a tuple $e = (u, v, t, \lambda)$, 
where $\tail{e} = u$ and $\head{e} = v$ denote the tail and head vertex, respectively, $t(e) = t$ denotes the \emph{time label}, and $\lambda(e) = \lambda$ denotes the \emph{traversal time} of $e$.
The time label is a point in time, after which $e$ is unavailable.
The traversal time is the time span it takes to travel along the temporal arc $e$.
If a temporal arc $e$ is \emph{delayed} by $\delta$, its time label increases by $\delta$.
A temporal graph $G$ is \emph{$(D, \delta)$-delayed}, for $D \subseteq E$, denoted as $G_{(D, \delta)}$, if the temporal arcs in $D$ are delayed by $\delta$.
Hence, $G_{(D, \delta)}$ is the $(D, \delta)$-delayed temporal graph version of the temporal graph $G$, where we replace each temporal arc $e \in D$ with $e' = (\tail{e}, \head{e}, t(e) + \delta, \lambda(e))$.
In the literature this type of delay is usually denoted as \emph{starting delay}~\cite{FuchsleMNR22STACS}.
For an arc $e$ in a temporal graph $G$ we denote by $\arrivalingraph{e}{G} = t(e) + \lambda(e)$ the \emph{arrival time} of $e$ (at vertex \head{e}) in the temporal graph $G$, \ie, $t(e)$ and $\lambda(e)$ are w.r.t.~the temporal graph $G$.
We omit $G$ and just write \arrival{e} if $G$ is clear from context.
A \emph{temporal walk} $W = (v_1, \dots, v_k)$ in $G$ is a walk where we require in addition that for each temporal arc $e_i = (v_i, v_{i+1}, \cdot, \cdot)$, $1 \leq i < k$, in $W$ we have $\arrival{e_i} \leq t(e_{i + 1})$. 
Similar to the static version, we define a \emph{temporal path} to be a temporal walk with pairwise distinct vertices.

Throughout this paper, we assume $t(e) > 0$ and $\lambda(e) > 0$ for all temporal arcs $e$ and write $t$ and $\lambda$ if $e$ is clear from context.
Furthermore, we drop the prefix ``temporal'' if there is no risk of confusion. 
Unless otherwise stated, we denote with $n$ the number of vertices and with $m$ the number of arcs of the graph $G$, \ie, $n = \vert V \vert$ and $m = \vert E \vert$.

\subsection{The \gamenamelong}
\label{sec:game-definition}
\looseness=-1
We now define the \emph{robust connection game}.
This is a round-based game between a \emph{traveler} and an \emph{adversary}.
Our goal is to model an online planning scenario, where delays are decided on by the adversary only on arrival of the traveler at a vertex and the traveler may reconsider his next steps.
An instance of robust connection game is given by a tuple $(G, s, z, x, \delta)$ where $G = (V, E)$ is a graph, $s, z \in V$, and $x$ and $\delta$ are positive integers.
The meaning is as follows:

\looseness=-1
The traveler starts at the vertex $s$ at timestamp $1$ and wants to reach the vertex $z$ in finitely many rounds.
Initially, we have $D = \emptyset$ and the adversary has a budget of $x$.
In each round, the traveler is located at a vertex $u \in G_{(D, \delta)}$ at some timestamp $t$ and the adversary first announces the delay of a subset $D'$ of the arcs.
The delayed arcs are restricted to those arcs that have $u$ as its tail and a time label $t'$ with $t \leq t'$ (arcs with $t > t'$ could be allowed, however the adversary does not gain an advantage from delaying such arcs).
Furthermore, the number of arcs that can be announced as delayed is limited to the remaining budget $x - \vert D\vert$ of the adversary.
The budget of the adversary is decreased after each round by the number of arcs that are announced as delayed, \ie, by $\vert D' \vert$.
Once the delays have been announced, a delay of $\delta$ time steps is applied to the (newly) delayed arcs, that is, we compute $G_{(D \cup D', \delta)}$.
Afterwards, the traveler has to move to a next vertex $v \neq u$, either through a delayed or not delayed arc $e$, 
setting the current time $t$ to $\arrivalingraph{e}{G_{\left(D \cup D', \delta\right)}}$.
Once the traveler is at $v$, the next round begins, \ie, the adversary can again announce the delay of a set $D'$ of arcs.

\looseness=-1
Each arc can be delayed at most once, \ie, once an arc has been announced as delayed and the delay was applied, it can never be re-delayed again.
The game ends if the traveler either reaches $z$, or if the traveler is stuck at a vertex $u \neq z$, that is, they reach $u$ at a time $t$ where there is no further arc $e$, \st, $t\leq t(e)$.
In the former case, the traveler wins the game, in the later case the adversary wins.
The traveler has a \emph{winning strategy}, if the vertex $z$ can always be reached independent of the announced delays.
We define the corresponding decision problem as follows.
\prob{
\gamenamelong~(\gamename)
}{
Input:
}{
A temporal graph $G = (V, E)$, two vertices $s, z \in V$, and two positive integers $x, \delta$, with $x \leq |E|$.
}{
Question:
}{
Does the traveler have a winning strategy in the robust connection game
$(G, s, z, x, \delta)$?
}{}

\section{Solving \gamenamelong{} in Exponential-Time and Polynomial-Space}
\label[section]{sec:PSPACE-containment}
\appendixsection{sec:PSPACE-containment}

%

Let $I = (G, s, z, x, \delta)$ be an instance of \gamenamelong~(\gamename).
We provide in this section an exponential-time dynamic programming (DP) algorithm to check whether the traveler has a winning strategy in $I$.
Füchsle et al.~\cite{FuchsleMNR22SAND} presented a polynomial-time DP-algorithm for the related \probname{Delayed-Routing Game}.
While our result follows a similar structure, the \gamenamelong{} does not possess polynomially many game states, which will also reflect in the running time.

For an instance $I$ of \gamenamelong\ as above we describe a state of the game with the tuple $\left(v, t, D\right)$, where $v \in V$ is the vertex the traveler is currently at, $t \in \mathbb{N}$ is the current timestamp, and $D \subseteq E$ denotes the set of delayed arcs.
The remaining budget of the adversary equals to $y := x - \vert D \vert$.
Observe that we do not have to consider all possible points in time but can restrict our attention to the set $T := \{1, \arrival{e}, \arrival{e} + \delta \mid e \in E)\}$ of all (delayed) arrival times at vertices.
In the following, we make use of the concept of (delayed) arcs available at a vertex $u \in V$ at the timestamp $t\in T$.
For the set of arcs $E$ in a $(D,\delta)$-delayed graph $G_{\left(D, \delta\right)}$, they are defined as~$E^t_{G_{\left(D, \delta\right)}}(u) := \{e \in E \mid \tail{e}= u, t \leq t(e)\}$.
Since $\delta$ is a fixed value, we simply write $E^t_D$. 

To solve instance $I$, for each of the possible game states $\left(v, t, D\right)$, we denote in a table $F$ whether the traveler has a winning strategy (\texttt{true}), or not (\texttt{false}).
\cref{lem:dp-recurrence relation} describes how the states of our game depend on each other.
\begin{lemma}
\label[lemma]{lem:dp-recurrence relation}
Assuming that the empty disjunction evaluates to \texttt{false}, we have the following equivalences for all $v \in V\setminus \{z\}$, $t \in T$, and $X \in \{D \mid D \subseteq E, \vert D \vert \leq x\}$.
\begin{align}
    F\left(z, t, X\right) &= \texttt{true}\label{eq:dp-recurrence-z}\\
    F\left(v, t, D\right) &= \bigwedge_{\substack{D' \subseteq E^t_{D}(v) \setminus D \\\text{\st}\ \vert D'\vert + \vert D \vert\leq x}}
    \quad\bigvee_{e \in E^t_{D\cup D'}(v)} F\left(\head{e}, \arrivalingraph{e}{G_{\left(D \cup D', \delta\right)}}, D \cup D'\right)\label{eq:dp-recurrence-not-z}
\end{align}
\end{lemma}
\begin{proof}
We first observe in \cref{eq:dp-recurrence-z} that the traveler is at the destination vertex $z$, \ie, the game is over and the traveler wins the game.
Hence, \cref{eq:dp-recurrence-z} is trivially correct.
We proceed with showing the correctness of \cref{eq:dp-recurrence-not-z}.
To do that, we show both directions explicitly and build our arguments on the definition of the game given in \cref{sec:game-definition}.

Assume that the traveler is at timestamp $t$ at vertex $v \neq z$ and the arcs $D \subseteq E$ have already been delayed by the adversary.
Therefore, we are in the game state $\left(v, t, D\right)$.
Furthermore, assume that the traveler has a winning strategy in this state, \ie, $F\left(v, t, D\right)$ is $\texttt{true}$.
This means that no matter which additional arcs the adversary decides to delay, the traveler can reach $z$.
So assume that the adversary announces the delay of the arcs in $D'$, which by our definition must have its tail at $v$.
A winning strategy at $\left(v, t, D\right)$ consists, by the definition of \gamenamelong, of using one outgoing arc $e$ of $v$, available at $t$, after being aware of the additionally delayed arcs.
However, observe that if the winning strategy at $\left(v, t, D\right)$ consists of moving in the presence of the additional delays $D'$ to the vertex $u = \head{e}$, for some arc $e$ available at $v$ at $t$, then the traveler must have a winning strategy in the state $\left(\head{e}, \arrivalingraph{e}{G_{\left(D \cup D', \delta\right)}}, D \cup D'\right)$.
That is, $F\left(u, \arrivalingraph{e}{G_{\left(D \cup D', \delta\right)}}, D \cup D'\right)$ must have been \texttt{true} as well.
Since the set of announced delayed arcs $D'$ was chosen arbitrarily w.r.t.~the constraints enforced by the game, the right side of \cref{eq:dp-recurrence-not-z} therefore correctly evaluates to true.
\toappendix{
Now assume that the right hand side of \cref{eq:dp-recurrence-not-z} is true and we are again in the game state $\left(v, t, D\right)$. 
Then for every possible subset $D'$, the expression $\bigvee_{e \in E^t_{D\cup D'}(v)} F(\head{e}, \arrivalingraph{e}{G_{\left(D \cup D', \delta\right)}}, D \cup D')$ is true.
In particular this means, that there exists an arc $e$, such that there is a winning strategy in the game state $(\head{e}, \arrivalingraph{e}{G_{\left(D \cup D', \delta\right)}}, D \cup D')$.
Thus for every possible subset $D'$ there exists an arc $e$ to traverse for the traveler, \st~they have a winning strategy in the resulting game state.
In other words, they have a winning strategy in $(v, t, D)$ and $F(v,t,D)$ is true.
}                               
\end{proof}

The game starts in the state $\left(s, 1, \emptyset\right)$ and, as a consequence of \cref{lem:dp-recurrence relation}, the traveler has a winning strategy iff $F\left(s, 1, \emptyset\right) = \texttt{true}$.
In the following, we derive the time required to compute $F\left(s, 1, \emptyset\right)$ as well as the space consumption of our approach.
\begin{lemma}
  \label[lemma]{lem:dp-algorithm-running-time}
  \looseness=-1
  Our approach solves \gamename{} in $\bigO(n \cdot m^2 \cdot \left(m + 1\right)^{2x})$ time.
\end{lemma}
\appendixproof{lem:dp-algorithm-running-time}{
  \looseness=-1
    We first derive the number of game states that we have to evaluate.
    Clearly, $\vert T \vert = \bigO(m)$ holds.
    For each $0 \leq y \leq x$, the number of possible sets $D \subseteq E$ of size $y$, that describe the delayed arcs, is $\binom{\vert E \vert}{y}$.
    In total, this is equal to
    \begin{align*}
        \sum_{y = 0}^x \binom{m}{y} \leq \left(m + 1\right)^x,
    \end{align*}
    where the upper bound follows from the fact that we can add a dummy element that will be selected the remaining $x - y$ times.
    Therefore, there are $\bigO(n \cdot m \cdot \left(m + 1\right)^x)$
    possible game states.

    \looseness=-1
    In each such game state $\left(v, t, D\right)$, we evaluate up to $\left(m + 1\right)^x$ possible sets $D'$ that the adversary could choose.
    For each such set $D'$ we look up $\vert E^t_{D\cup D'}(v) \vert = \bigO(m)$ other table entries.
    Note that, for evaluating \arrivalingraph{e}{G_{\left(D \cup D', \delta\right)}} we do not need to compute the $\left(D \cup D', \delta\right)$-delayed graph, but can just check whether $e \in D \cup D'$ holds or not.
    Therefore, we can determine whether the traveler has a winning strategy in the state $\left(v, t, D\right)$ in $\bigO(m\cdot \left(m + 1\right)^x)$ time.
    Combining the above yields the claimed running time of $\bigO(n \cdot m^2 \cdot \left(m + 1\right)^{2x})$.
  } 
\begin{lemma}
    \label[lemma]{lem:dp-algorithm-space}
    Our approach for \gamename{} can be implemented to use $\bigO(x\cdot m)$ space.
\end{lemma}
\begin{proof}
    As we show in the proof of \cref{lem:dp-algorithm-running-time}, there are $\bigO(n \cdot m \cdot \left(m + 1\right)^x)$ possible game states.
    Na\"ively enumerating all possible game states would thus require an exponential amount of space.
    To circumvent this, we first observe that, by \cref{eq:dp-recurrence-not-z}, evaluating whether the traveler has a winning strategy in the initial game state $\left(s, 1, \emptyset\right)$ results in a search tree $\mathcal{T}$, in which we enumerate in every odd level of $\mathcal{T}$ all possible subsets $D'$ of delayed arcs, and in every even level of $\mathcal{T}$ the possible arcs the traveler can use.
    If we now use depth first search (DFS) on $\mathcal{T}$ to compute $F\left(s, 1, \emptyset\right)$, we only have to store the states of a single path from the root of $\mathcal{T}$ to a leave of $\mathcal{T}$ at a time.
    If we fix the order in which we enumerate at each internal node of $\mathcal{T}$ its children,
    the space requirement for DFS is linear in the depth of $\mathcal{T}$.
    We conclude the proof with observing that the depth of $\mathcal{T}$ is $\bigO(\vert T \vert) = \bigO(m)$, since we increase at every other level in $\mathcal{T}$ the timestamp.
    Each game state requires $\bigO(x)$ space, since we need to store, in the worst case, that many delayed arcs.
\end{proof}
Using \cref{lem:dp-algorithm-running-time,lem:dp-algorithm-space} we summarize the main result of this section in \cref{thm:dp-algorithm}.
\begin{theorem}
    \label[theorem]{thm:dp-algorithm}
    \looseness=-1
    Let $I = (G = (V, E), s, z, x, \delta)$ be an instance of \gamenamelong{} with $n = \vert V \vert$ and $m = \vert E \vert$.
    We can solve 
    $I$ in $\bigO(n \cdot m^2 \cdot \left(m + 1\right)^{2x})$ time using $\bigO(x\cdot m)$ space. Therefore, \gamenamelong\ is in \PSPACE.
\end{theorem}

\section{\gamenamelong{} is \PSPACE-hard}
\label[section]{sec:PSPACE-hardness}
\appendixsection{sec:PSPACE-hardness}

A quantified boolean formula (QBF) is a formula $\phi = Q_1x_1Q_2x_2 \dots Q_nx_n\varphi$ with $Q_i \in \{\exists, \forall\}$. Deciding satisfiability of QBFs is well known to be \PSPACE-complete. We show \PSPACE-hardness by providing a reduction from deciding satisfiability of QBFs. Similar to our problem, deciding satisfiability can be formulated as finding a winning strategy for a 2-player game. In the \probname{QBF Game}, we have an $\exists$-player and a $\forall$-player. For each $Q_i$, if $Q_i=\exists$ the $\exists$-player selects a truth assignment for $x_i$. Otherwise, if $Q_i=\forall$ the $\forall$-player selects a truth assignment for $x_i$. The $\exists$-player wins if after the $n$-th round $\varphi$ has a satisfying assignment. A QBF $\phi$ is satisfiable if and only if the $\exists$-player has a winning strategy \cite{shukla_survey_2019}.

We provide a reduction from \probname{QBF Game} to \gamenamelong{}. Let $\phi = Q_1x_1Q_2x_2 \dots Q_nx_n\varphi$ be a QBF. We can safely assume that $\varphi$ is in conjunctive normal form, \ie,~$\varphi = C_1 \land \dots \land C_m$. We construct an instance $I = (G, s, z, x, \delta)$ of the \gamenamelong{} from $\phi$.
In particular we will create $G$ based on the order and type of quantifiers in $\phi$ by placing a gadget (a specific subgraph) for every such quantifier.
The travel time of every arc $e$ in $G$ is uniformly set to $\lambda(e)=1$.
Additionally, we specifiy a start vertex $s$ and an end vertex $z$ in $G$.
Finally we set $\delta = 1$ and the budget $x=n(m + |\forall|) + 1$, where $|\forall|$ is the number of universally quantified variables.
The significance of the value of $x$ will become obvious with the construction of $G$. 


In this section we will first describe the structure of the created instance (in particular the construction of $G$).
We also state some observations about the forced behavior of the traveler and the adversary when playing the \gamenamelong{} on $I$.
Finally, we will describe the equivalences between the choices of the traveler/the adversary and the $\exists$-player/the $\forall$-player in the \probname{QBF Game}.

\subsection{Construction of the graph $G$}
\looseness=-1
First let  $Q_i$ be the $i$-th quantifier, quantifying the variable $x_i$.
We place a gadget for $Q_i$, which is a specific subgraph of $G$.
If $Q_i$ is an existential quantifier we place the existential gadget $\mathcal{G}^\exists_i$, otherwise we place the universal gadget $\mathcal{G}^\forall_i$, both defined below.
Every such gadget has an entry vertex $s_i$ and an exit vertex $s_{i+1}$, \ie, the exit vertex of a gadget $\mathcal{G}_i$ is the entry vertex of the next gadget $\mathcal{G}_{i+1}$, with the exception of $\mathcal{G}_n$, whose exit vertex $s_{n+1}$ is not an entry vertex, since no following gadget exists.
The starting time of all outgoing arcs of $s_i$ will be identical and denoted as $t_{s_i}$.
We define $t_{s_i} = (i - 1) \cdot 2(m + 2)$.
Thus, the traveler needs to be able to traverse the variable gadget in $2(m+2)$ units of time or less to have a winning strategy, otherwise they arrive at the next gadget and can not continue on from $s_{i+1}$.
We will now first describe the construction of the two gadget types in detail, before carrying on with the further construction of $G$.

\begin{figure}[tbp]
    \centering	
    \begin{subfigure}[t]{.48\linewidth}
        \centering
		\includegraphics[width=\linewidth, page=6]{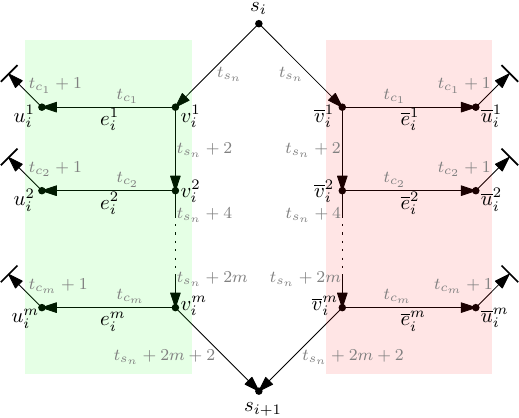}
		\caption{Existential variable gadget}
        \label{fig:exist_variable_gadget}
    \end{subfigure}%
    \hfill
    \begin{subfigure}[t]{.48\linewidth}
		\centering
		\includegraphics[width=\linewidth, page=7]{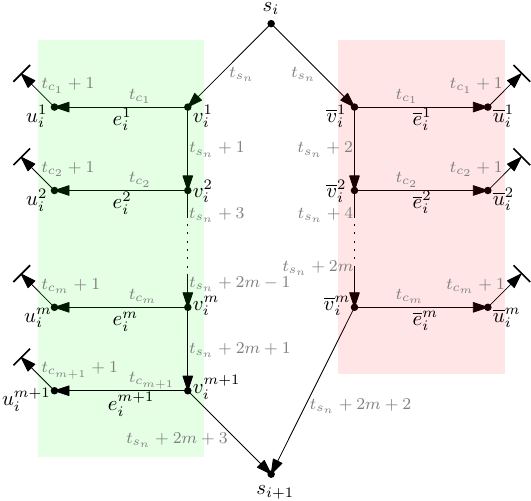}
		\caption{Universal variable gadget}
        \label{fig:universal_variable_gadget}
    \end{subfigure}
    \caption{Depiction of an existential and universal variable gadget. Escape arcs to $z$ are depicted as upwards pointing arcs.}
    \label{fig:variable_gadget}
\end{figure}

\paragraph{Existential variable gadget.}
We begin with the existential variable gadget $\mathcal{G}^\exists_i$ (see \cref{fig:exist_variable_gadget}).
For each clause $C_j$, the gadget contains two vertex pairs $u_i^j, v_i^j$ and $\overline{u}_i^j, \overline{v}_i^j$ connected via arcs $e_i^j=(u_i^j, v_i^j, t_{c_j}, 1)$ and $\overline{e}_i^j = (\overline{u}_i^j, \overline{v}_i^j, t_{c_j}, 1)$ respectively.
We will call $e_i^j$ a \emph{positive} and $\overline{e}_i^j$ a \emph{negative} arc.
It is important to mention that $t_{c_j} > t_{s_{n+1}}$, \ie, any positive or negative arc of any gadget has a leaving time later than the outgoing arcs of $s_{n+1}$.

Furthermore, there are two distinct paths from $s_i$ to $s_{i+1}$, namely $P_i = \langle s_i, u_i^1, \dots, u_i^m, s_{i+1} \rangle$ and $\overline{P}_i = \langle s_i, \overline{u}_i^1, \dots, \overline{u}_i^m, s_{i+1} \rangle $.
The traveler chooses one of these two paths by choosing which of the two outgoing arcs at $s_i$ they traverse.
The starting times of two consecutive arcs of both $P_i$ and $\overline{P}_i$ are always two time units apart, thus the traveler requires at most $2(m+1)$ units of time to travel from $s_i$ to $s_{i+1}$.
Lastly, at each vertex $u_i^j, \overline{u}_i^j$ there will be an arc to $z$, starting at timestamp $t_{c_j}+1$.
We will refer to such arcs as \textit{escape arcs}.

\begin{observation}\label{obs:path_traversal}
	If the traveler arrives at $s_i$ of $\mathcal{G}^{\exists}_i$ at a timestamp $t\leq t_{s_i}$, the adversary cannot prevent the traveler from arriving at $s_{i+1}$ at a timestamp $t' \leq t_{s_{i+1}}$ by delaying any or all arcs of $P_i$ and $\overline{P}_i$.
\end{observation}
\begin{proof}
	The traveler can clearly take either of the two arcs $(s_i, v^1_i), (s_i,\overline{v}^1_i)$, if they are delayed or not.
	Since the arrival time of any arcs on the paths is always two units before the leaving time of the next and $\delta=1$, the observation follows.
\end{proof}


\begin{observation}\label{obs:edge_delays}
    \looseness=-1
	If the traveler arrives at a vertex $v^j_i$ of $\mathcal{G}^{\exists}_i$, they can reach $u^j_i$ at a timestamp $t\leq t_{c_j}+1$ and subsequently reach $z$ if and only if, the adversary does not delay $e^j_i$.
\end{observation}
\begin{proof}
	We can assume that the escape arc $a$ at $u^j_i$ is not delayed since delaying an escape arc can never prevent the traveler from reaching $z$.
	The observation immediately follows from the fact that $t(e^j_i) + \lambda(e^j_i) = t(a)$ if $e^j_i$ is not delayed and $t(e^j_i) + \lambda(e^j_i) > t(a)$ if it is.
\end{proof}

The following \cref{lem:existential-traversal} summarizes the properties of the gadget, taking into account Observation~\ref{obs:path_traversal} and Observation~\ref{obs:edge_delays}.
\begin{lemma}
\label[lemma]{lem:existential-traversal}
  If the traveler arrives at $s_i$ of $\mathcal{G}^{\exists}_i$ at a timestamp $t\leq t_{s_i}$, either they can reach $z$ or they arrive at $s_{i+1}$ at a timestamp $t' \leq t_{s_{i+1}}$ via $P_i$ (or $\overline{P}_i$) and the adversary has delayed all $m$ arcs $e_i^j$ (or all $m$ arcs $\overline{e}_i^j$).
\end{lemma}
\appendixproof{lem:existential-traversal}{%
    \looseness=-1
	By Observation~\ref{obs:path_traversal}, we know that the arrival at $s_{i+1}$ can not be prevented.
	The traveler has to start by taking one of the two arcs $(s_i, v^1_i), (s_i,\overline{v}^1_i)$.
	Assume, \wlogg, they choose $(s_i, v^1_i)$.
	From there by Observation~\ref{obs:edge_delays} we know that either the traveler can reach $z$ or the adversary delays an arc $e_i^j$.
	This repeats $m$ times in total until the traveler reaches $s_{i+1}$.
} 


\looseness=-1
\paragraph{Universal variable gadget.}
The universal variable gadget $\mathcal{G}^{\forall}_i$ has a very similar structure to the existential variable gadget with only two minor modifications (see \cref{fig:universal_variable_gadget}).
First, we add two additional vertices $u_i^{m+1}$ and $v_i^{m+1}$ which are connected by an arc $e_i^{m+1} = (u_i^{m+1}, v_i^{m+1}, t_{c_{j+1}}, 1)$.
Again, $t_{c_{j+1}} > t_{s_{n+1}}$ holds.
There will be an escape arc from $v_i^{m+1}$ leaving at timestamp $t_{c_{j+1}}+1$.
Furthermore, the node $u_i^{m+1}$ is inserted into the path $P_i$ after node $u_i^m$.
The starting times of the arcs along path $P_i$ are adapted accordingly to guarantee the two time unit difference between two consecutive arcs along the path.
Secondly, the starting times for all arcs along path $P_i$, except for the first arc, are decreased by one unit.
Thus, the traveler requires at most $2(m+1) + 1$ units of time to travel from $s_i$ to $s_{i+1}$ along path $P_i$.
Everything else remains the same.

Note that while Observation~\ref{obs:edge_delays} directly translates to $\mathcal{G}^{\forall}_i$, Observation~\ref{obs:path_traversal} does not.
Instead we make the following observation.

\begin{observation}\label{obs:forall_path_traversal}
	If the traveler arrives at $s_i$ of $\mathcal{G}^{\forall}_i$ at a timestamp $t\leq t_{s_i}$, the adversary can prevent the traveler from traversing $P_i$ by delaying the arc $(s_i, v^1_i)$.
\end{observation}
 \ifarxiv
\begin{proof}
	This follows directly from $t(s_i, v^1_i) + \lambda(s_i, v^1_i) > t(v^1_i, v^2_i)$ if $(s_i, v^1_i)$ is delayed.
\end{proof}
\fi

Note that the traversal of $\overline{P}_i$ can still not be impeded by the adversary.


Similarly to the existential variable gadget, the traveler will either reach $z$ or $s_{i+1}$ after traversing the gadget. However, in the latter case, three different outcomes may arise with respect to the number of delayed edges. Note that one of the possibilities (Situation b in Lemma~\ref{lem:universal-traversal}) will lead to a forced loss of the traveler later on.
We will argue this after presenting all gadgets. 

\begin{lemma}
  \label[lemma]{lem:universal-traversal}
  \looseness=-1
  If the traveler arrives at $s_i$ of $\mathcal{G}^{\forall}_i$ at a timestamp $t\leq t_{s_i}$, either they can reach $z$ or the traveler reaches $s_{i+1}$ at timestamp $t' \leq t_{s_{i+1}}$ and one of the following situations occurs:
  %
  \begin{compactenum}
  \item[(a)] All $m+1$ positive arcs $e_i^j$ are delayed.
  \item[(b)] All $m$ negative arcs $\overline{e}_i^j$ are delayed.
  \item[(c)] All $m$ negative arcs $\overline{e}_i^j$ as well as the first arc in $P_i$ are delayed.
  \end{compactenum}
\end{lemma}
\appendixproof{lem:universal-traversal}{%
  At $s_i$ the adversary has the option to delay the first arc in $P_i$.
  If the adversary does not delay the first arc of $P_i$, the traveler can choose between path $P_i$ and $\overline{P}_i$.
  If they choose $P_i$ they would arrive at $s_{i+1}$ at timestamp $t_{s_i} + 2 m + 2$, delaying all remaining $m+1$ arcs $e_i^j$ (Situation a).
  If they instead choose path $\overline{P}_i$, they would arrive at timestamp $t_{s_i} + 2 m + 1$, delaying all $m$ arcs $\overline{e}_i^j$ (Situation b).
  If the adversary does delay the first arc of $P_i$, by Observation~\ref{obs:forall_path_traversal}, the traveler can not travel along $P_i$.
  Thus, the traveler would go along $\overline{P}_i$ and arrive at $s_{i+1}$ at timestamp $t_{s_i} + 2 m + 1$, delaying all $m$ arcs $\overline{e}_i^j$ (Situation c).
  We observe that regardless of the path choice, the adversary will arrive at $s_{i+1}$ before $t_{s_{i+1}}$.
}

The construction of $G$ and the observed traversal behavior of the two players maps naturally to an assignment of a variable $x_i$ to either true or false if the set of all positive or all negative arcs in $\mathcal{G}_i$ is delayed, respectively.
We call the traversal of the traveler from $s$ until $s_{n+1}$, the \emph{assignment phase}.
Now we continue with the remaining construction of $G$, which will model the evaluation of this variable assignment.
The remaining traversal of the traveler starting at $s_{n+1}$ will accordingly be called the \emph{evaluation phase} (see \cref{fig:reduction_timeline,fig:reduction_example}).

\begin{figure}[tbp]
    \centering
    \includegraphics[page=2]{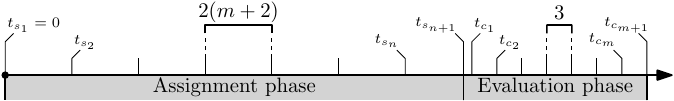}
    \caption{Schedule of the reduction}
    \label{fig:reduction_timeline}
\end{figure}

In the evaluation phase, we check whether the variable assignments implied by the choices of the traveler and the adversary result in each clause $C_j$ in $\varphi$ to evaluate to true.
This is done by placing a gadget $\mathcal{G}^C_j$ for every clause $C_j$.
The gadget $\mathcal{G}^C_j$ contains an entry node~$c_j$ and an exit node $c_{j+1}$ (which is again the entry node for $\mathcal{G}^C_{j+1}$).
The gadgets have to be traversed in sequential order.
The starting time of all outgoing arcs of $c_j$ is identical and denoted as $t_{c_j}$.
We define $t_{c_j}=t_{s_{n+1}} + 1 + 3 (j - 1)$ (three time units per clause).



The last vertex of the assignment phase, $s_{n+1}$, is connected to $c_1$ by an arc with starting time $t_{s_{n+1}}$.
Furthermore, there will be an additional arc leaving $s_{n+1}$ to new vertex $v$, also leaving at timestamp $t_{s_{n+1}}$.
Vertex $v$ is connected by an escape arc to vertex $z$ with starting time $t_{s_{n+1}} + 1$.
Lastly, the final exit node $c_{m+1}$ of the last clause gadget is connected to target vertex $z$ with starting time $t_{c_{m+1}}$.

\begin{lemma}
  \label[lemma]{lem:delay_budget_0}
  When the traveler reaches $c_1$, if the adversary has delayed less than $n(m+|\forall|) + 1$ arcs, where $|\forall|$ is the number of universally bound variables, the traveler was able to reach $z$ in the assignment phase.
\end{lemma}
\appendixproof{lem:delay_budget_0}{%
  For all existential gadgets, by \cref{lem:existential-traversal}, $m$ arcs have been delayed or the traveler was able to reach $z$.
  For all universal gadgets, according to \cref{lem:universal-traversal}, either $m+1$ (Situations a and c) or $m$ arcs (Situation b) have been delayed, or the traveler was able to reach $z$.
  If for all universal gadgets Situations a or c occurred, we are done.
  Assume that Situation b occurred at least once.
  In that case, less than $n(m+|\forall|)$ arcs have been delayed when the traveler arrives at node $s_{n+1}$.
  The adversary can now delay both outgoing arcs from $s_{n+1}$, leaving the adversary either stranded at vertex $v$ or $c_1$.
  Choosing Situation b therefore loses the game for the traveler.
  Thus the traveler will never choose Situation b and exactly $n(m+|\forall|)$ arcs have been delayed in the assignment phase.
  
  This leaves one delay remaining when the traveler arrives at node $s_{n+1}$.
  If adversary does not delay the arc to vertex $v$, the traveler can move to $v$ and take the escape arc.
  Therefore the adversary spends the remaining budget on delaying the arc $(s_{n+1},v)$, the traveler moves to vertex $c_1$ and $n(m+|\forall|) + 1$ arcs have been delayed.
} 

Further if the game progresses to a state, in which the traveler reaches $c_1$ and was so far unable to reach $z$, the budget of the adversary must be 0 and due to Lemma~\ref{lem:delay_budget_0} we will assume that this is the case for the remainder of this reduction.

\begin{figure}[tbp]
	\centering
	\includegraphics[width=.55\linewidth, page=8]{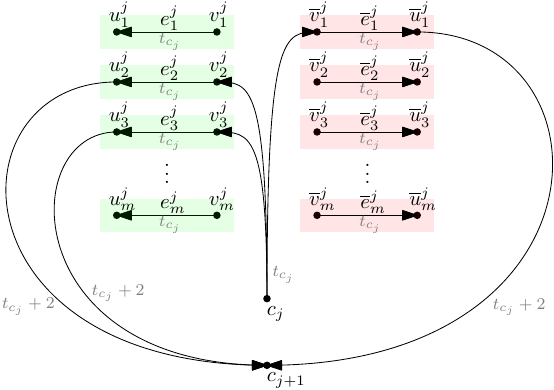}
	\caption{Clause gadget}
	\label{fig:clause_gadget}
\end{figure}

\paragraph{Clause gadget.}
\looseness=-1
The clause gadget for some clause $C_j$ has an arc from $c_j$ to $v_i^j$ (or $\overline{v}_i^j$) for each literal $x_i$ (or $\neg x_i$) in $C_j$ with starting time $t_{c_j}$.
Equivalently, there is an arc from $u_i^j$ (or $\overline{u}_i^j$) to $c_{j+1}$ for each literal $x_i$ (or $\neg x_i$) with starting time $t_{c_j} + 2$.
This forms a set of parallel paths, each containing some arc $e_i^j$ (or $\overline{e}_i^j$).
We will call the the set of arcs $e_i^j$ (or $\overline{e}_i^j$) the \emph{literal arcs} of $\mathcal{G}^C_j$.
Note that the set of literal arcs can include both positive and negative arcs of the variable gadgets.
The gadget is depicted in \cref{fig:clause_gadget}.
Also note that if an arc $a = (u_i^j, c_{j+1})$ (or $(\overline{u}_i^j, c_{j+1})$) is delayed, the traveler will be stuck at $c_{j+1}$.
A direct consequence of Lemma~\ref{lem:delay_budget_0} is that at the adversary will always have at least $1$ delay remaining in the assignment phase.
This prevents the traveler from traversing a literal arc during the assignment phase and using $a$ to ``jump ahead'' to the evaluation phase, since the adversary could delay $a$.
Conversely, in the evaluation phase, the adversary has no budget left and cannot delay $a$.


\begin{lemma}
  \label{lem:literal_arc_clause_traversal}
  \looseness=-1
  If the traveler arrives at the entry node $c_j$ of $\mathcal{G}^C_j$ at a timestamp $t \le t_{c_j}$ they can travel to the exit node $c_{j+1}$ of $\mathcal{G}^C_j$ if and only if at least one literal arc of $\mathcal{G}^C_j$ has been delayed.
\end{lemma}
\appendixproof{lem:literal_arc_clause_traversal}{
	In order to travel from $c_j$ to $c_{j+1}$, the traveler has to pick one of the parallel paths.
	\Wlogg~suppose there exits a literal arc $e_i^j$ of $\mathcal{G}^C_j$ which has been delayed.
	Then, the traveler chooses the path that contains this arc and arrives at $v_i^j$ at timestamp $t_{c_j} + 1$.
	Since $e_i^j$ has been delayed, instead of departing at timestamp $t_{c_j}$ it leaves at timestamp $t_{c_j}+1$, making the arc traversable for the traveler.
	When the traveler arrives at vertex $u_i^j$, it can move to vertex $c_i^j$ and arrive there at timestamp $t_{c_j}+3 \le t_{c_{j+1}}$.
	Now suppose that no literal arc is delayed.
	Assume the traveler picks an arbitrary path.
	\Wlogg~this path contains the literal arc $e_i^j$.
	When the traveler arrives at vertex $v_i^j$ at timestamp $t_{c_j} + 1$, there is no arc leaving $v_i^j$ that the traveler can reach in time since $e_i^j$ is not delayed and already departs at timestamp~$t_{c_j}$.
	The traveler would thus be stuck at vertex $v_i^j$ and never reach $c_{j+1}$.
} 

\begin{figure}[tbp]
	\centering
	\includegraphics[page=10, width=\linewidth]{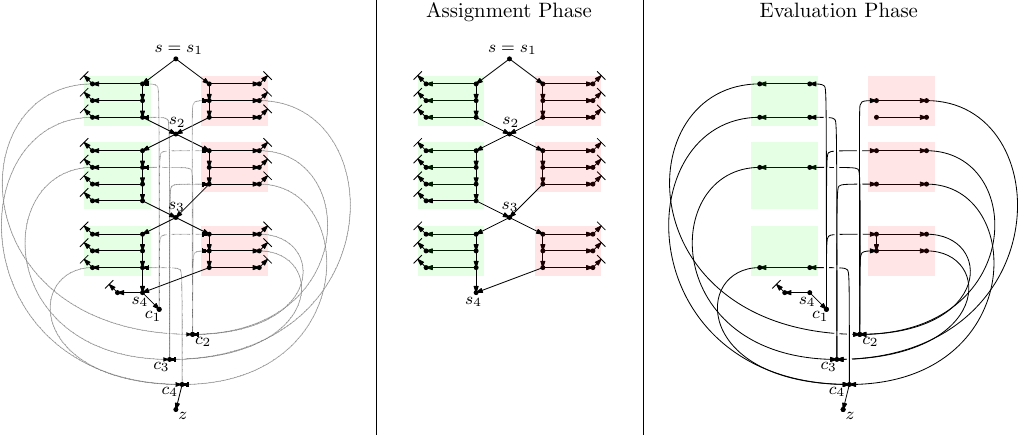}
	\caption{Resulting graph $G$ from reducing formula $\phi = \exists x_1 \forall x_2 \exists x_3 ((x_1 \lor \lnot x_2 \lor \lnot x_3) \land (\lnot x_1 \lor x_2 \lor \lnot x_3) \land (x_1 \lor \lnot x_2 \lor x_3))$. On the right you see the two phases extracted from $G$.}
	\label{fig:reduction_example}
\end{figure}

With the construction of $G$ and therefore the whole instance $I$ concluded, we can move on to the main statement of the reduction.

\subsection{Proving \PSPACE-completeness}

Here we establish the connection and in fact the equivalence of a the graph traversal in the \gamenamelong\ and the variable assignments in the \probname{QBF Game}.

\begin{lemma}
  \label[lemma]{lem:qbf-equal-robust-connection}
  \looseness=-1
  There is a winning strategy for \probname{QBF Game} for $\phi$ if and only if there is a winning strategy for the traveler in the reduction instance of \gamenamelong{}.
\end{lemma}
\appendixproof{lem:qbf-equal-robust-connection}{
	\begin{description}
		\item[$\Rightarrow$:]
		Assume the $\exists$-player has a winning strategy for the \probname{QBF Game}.
		In particular this means that there is a decision tree for the $\exists$-player to choose a value for every existentially quantified variable, given the truth values of all previously quantified variables, \st, the resulting variable assignment satisfies $\varphi$.
		The corresponding winning strategy for the traveler in the \gamenamelong{} follows every decision in this decision tree at every $\mathcal{G}^{\exists}_i$ gadget.
		If the decision is to set the variable $x_i$ to true, the traveler chooses the $P_i$ path, otherwise they choose the $\overline{P}_i$ path.
		At every $\mathcal{G}^{\forall}_i$ gadget, if the adversary delays the first arc of the $P_i$ path, the traveler chooses the $\overline{P}_i$ path, otherwise they chooses the $P_i$ path.
		Since the variable assignment in the \probname{QBF Game} satisfies every clause in $\varphi$, at least one literal of the clause is true and therefore at least one literal arc of every clause gadget is delayed.
		By Lemma~\ref{lem:literal_arc_clause_traversal} the traveler can traverse every clause gadget and reach $z$.

                \looseness=-1
		\item[$\Leftarrow$:]
		Now assume the traveler has a winning strategy for the \gamenamelong.
		In particular this means, the traveler always chose the ``correct'' path in every $\mathcal{G}^{\forall}_i$ gadget, depending on the choice of the adversary to delay the first arc of $P_i$ or not.
		Moreover, for every such choice the traveler could choose one of the two available paths $P_i$ or $\overline{P}_i$ in every $\mathcal{G}^{\exists}_i$ gadget, \st, they could traverse every clause gadget to reach $z$.
		Since by Lemmas~\ref{lem:existential-traversal} and~\ref{lem:universal-traversal} only positive or negative arcs, but never a mixture of both can be delayed in every variable gadget, we can follow all decisions of the traveler and set the variable $x_i$ to true if the traveler would choose path $P_i$ and to false otherwise.
		Finally, since the traveler could traverse every clause gadget and by Lemma~\ref{lem:literal_arc_clause_traversal}, there is at least one delayed literal arc in every clause gadget. Therefore, the resulting variable assignment satisfies $\varphi$ and the $\exists$-player has a winning strategy. \qedhere
	\end{description}
} 
Since deciding if there is a winning strategy in the \probname{QBF Game} is \PSPACE-hard, we conclude the following theorem from \cref{lem:qbf-equal-robust-connection} and \cref{thm:dp-algorithm}.

%

\begin{theorem}
	\label[theorem]{thm:completeness}
	\gamenamelong{} is \PSPACE-complete.
\end{theorem}

\section{Conclusion}
\label{sec:conclusion}

\looseness=-1
We close with directions for future research:
First, the running time for \gamenamelong\
has the form $n \cdot m^{\bigO(x)}$, where $n$, $m$, and $x$ are the number of vertices, temporal arcs, and delayed arcs, respectively.
It is natural to ask whether we can remove the dependence of the exponent on~$x$ to obtain a running time of the form $f(x) \cdot (n \cdot m)^{\bigO(1)}$.
Furthermore, what if we ask for routes with only a small number $c$ of changeovers?
Is there a running time of $(n \cdot m)^{f(c)}$ possible?
\ifarxiv
Additionally, one could investigate a 'static' variant of \gamename:
The variant in which we ask whether there is a route that is feasible for every possible bounded-size set of delays has been studied by \citet{FuchsleMNR22STACS}.
However, the variant where we ask whether for every bounded-size set of delays there exists a feasible route is yet unstudied for our model.
\fi
Finally, it would be interesting to see how we can incorporate delays that vary between different routes or change over time into our model and whether this impacts our results.

\textbf{Acknowledgements.}
The authors would like to thank Manuel Sorge for his input on the hardness proof, the fruitful discussion, and the feedback on the first version of this paper.

\printbibliography

\appendix

\end{document}
